\documentclass[journal,comsoc]{IEEEtran}
\usepackage{cite}

\usepackage{amsmath}
\usepackage{amsthm}
\newtheorem{theorem}{Theorem}
\newtheorem{corollary}{Corollary}
\newtheorem{lemma}{Lemma}
\newtheorem{proposition}{Proposition}
\theoremstyle{definition}
\newtheorem{definition}{Definition}
\theoremstyle{remark}
\newtheorem*{remark}{Remark}

\usepackage{mathtools}
\newcommand{\numberthis}{\stepcounter{equation}\tag{\theequation}}
\newcommand{\labelsign}[2]{\overset{\mathclap{\textrm{(#1)}}}{#2}}

\usepackage[cmintegrals]{newtxmath}
\usepackage{enumitem}
\usepackage{tabu}
\usepackage{hyperref}

\usepackage{tikz}
\usepackage{tikzscale}
\usepackage{pgfplots}
\pgfplotsset{compat=1.17}
\usetikzlibrary{calc}
\usetikzlibrary{positioning}
\usetikzlibrary{intersections}
\usetikzlibrary{backgrounds}

\begin{document}

\title{Over-the-Air Statistical Estimation}

\author{Chuan-Zheng~Lee,
        Leighton~Pate~Barnes,
        and
        Ayfer~\"Ozg\"ur,~\IEEEmembership{Senior~Member,~IEEE}%
\thanks{The authors are with the Department of Electrical Engineering, Stanford
        University, Stanford, CA 94305 USA (email: czlee@stanford.edu;
        lpb@stanford.edu; aozgur@stanford.edu).}%
}

\maketitle

\begin{abstract}
    We study schemes and lower bounds for distributed minimax statistical estimation over a Gaussian multiple-access channel (MAC) under squared error loss, in a framework combining statistical estimation and wireless communication. First, we develop ``analog'' joint estimation-communication schemes that exploit the superposition property of the Gaussian MAC and we characterize their risk in terms of the number of nodes and dimension of the parameter space. Then, we derive information-theoretic lower bounds on the minimax risk of any estimation scheme restricted to communicate the samples over a given number of uses of the channel and show that the risk achieved by our proposed schemes is within a logarithmic factor of  these lower bounds. We compare both achievability and lower bound results to previous ``digital'' lower bounds, where nodes transmit errorless bits at the Shannon capacity of the MAC, showing that estimation schemes that leverage the physical layer offer a drastic reduction in estimation error over digital schemes relying on a physical-layer abstraction.
\end{abstract}

\begin{IEEEkeywords}
    federated learning, over-the-air learning, statistical estimation
\end{IEEEkeywords}

\section{Introduction}

To fully appreciate the plenitude of data fueling the modern rise of machine learning, we might pause to consider not just its volume, but its origins.  While the computational lifting is often concentrated in powerful central servers, the data they rely on is largely and increasingly born ``at the edge''---spawned in a myriad of devices, scattered ubiquitously, equipped with sensors and user input to collect all sorts of information from the world.

Recognizing this growing decentralization of data, there has been growing interest in the study of techniques to combine samples from many nodes to make inferences. The key distinction between this and traditional approaches to learning and estimation is an explicit consideration of the communication channels between edge devices and the central server, which are often noisy or unreliable channels, like wireless links.  Two disciplines with decades-long histories, wireless networks and statistical estimation, thus find themselves with common cause.

Recent years have seen significant activity in this nascent intersection. One simple and intuitive way to study the impact of bandwidth-limitations on estimation performance is to model them as imposing constraints on the number of bits available to encode each observed sample. A number of works in the machine learning literature have taken this tack \cite{duchi,garg,braverman,diakonikolas,archayaetal,barnes-lowerbounds}, providing both achievable schemes and information-theoretic lower bounds under bit constraints, and characterizing the dependence of the estimation error on the number of bits available to represent each sample. Since these analyses assume the encoded bits to be received without error, they can be interpreted as assuming that a reliable scheme is used for transmission over the underlying noisy channels. This in effect suggests an abstraction layer, separating the question of physical-layer communication from the statistical estimation problem.

Our goal in this paper is to study the problem of distributed statistical estimation over a noisy multiple-access channel from first principles. Our main contributions are as follows.

First, we introduce a model for minimax parameter estimation over a fixed number of uses of the Gaussian multiple-access channel (MAC), defining the estimation schemes of interest in this setting and providing a rigorous and tractable mathematical model for quantifying their performance and studying their optimality.

Second, we develop analog transmission-estimation schemes for two canonical estimation tasks, the Gaussian location and product Bernoulli models, in which each node scales and transmits its uncoded sample to the central server, leveraging the superposition of the Gaussian MAC to perform averaging over the air.  We analyze these schemes and characterize their risk under squared error loss. When compared to information theoretic \textit{lower} bounds for digital schemes, controlling for physical resources, we find that the worst-case risk of these analog schemes is exponentially smaller than that of any digital scheme. This suggests that analog schemes that consider estimation and physical layer transmission jointly can bring about drastic improvements over digital schemes that separate the two with an abstraction layer.

We next address the question of whether the analog schemes we develop are close to optimality. We derive a fundamental lower bound on the risk achievable by any estimation-communication scheme satisfying the physical constraints of our model. This bound uses the recent result of \cite{barnes-fisher-arxiv}, which showed an upper bound on the Fisher information of a channel's output in terms of the channel's Shannon mutual information when the statistical model being estimated exhibits a sub-Gaussian score. We apply this result to the aforementioned two estimation tasks, and find the risk achieved by our estimation schemes to be within a logarithmic factor of the lower bound.  To the best of our knowledge, this is the first information-theoretic lower bound for distributed minimax estimation over a noisy multi-user channel.

\subsection{Related works}

Our results corroborate similar gains from uncoded analog transmission that have been observed in source coding for sensor networks \cite{gastpar-TIT-2008,gastpar-IPSN-2003}. They also reinforce a number of recent works proposing adaptations of common learning algorithms in a wireless federated learning context. A notable example is distributed stochastic gradient descent, where leveraging the Gaussian MAC for analog averaging has been observed experimentally to far outperform digital approaches \cite{amiri-sgd-air,amiri-fedlearning-fading,amiri-convergence-of-downlink-arxiv}. On the theoretical side, analysis of methods using analog over-the-air gradient aggregation has shown convergence rates similar to error-free channels \cite{sery-gd-mac-fading,sery-cotaf-globecom2020,sery-heterogeneous-data-arxiv}, albeit without comparison to digital counterparts.

Other works with analog aggregration showing an advantage over digital methods include \cite{zhu-broadband-latency}, which replaced digital with analog modulation of model parameters to improve latency with comparable accuracy, and \cite{du-fast-analog-transmission}, which used an analog method with MIMO antennas. On the other hand, \cite{ahn-federated-distillation} found the digital and analog schemes it studied to perform comparably in numerical experiments. Such analog methods have seen wider development efforts \cite{yang-beamforming,guo-gradient-aggregation}, including with privacy considerations in mind \cite{seif-wireless-fl-analog-privacy,abdi-random-linear-coding,barnes-ldp,liu-privacy-free}. This idea has also been applied to over-the-air computation more generally \cite{wu-stac,liu-ota-scaling-laws,dong-blind-aircomp,zang-aircomp,chen-aircomp}. While the current literature supports the utility of analog schemes, none of these works provide an explicit analytical comparison from first principles between analog achievability and digital lower bounds for canonical estimation and inference tasks.

More broadly, there has been interest in federated learning over wireless networks from a range of angles. One idea to preserve bandwidth is to take advantage of gradient sparsity \cite{sun-sparsification,jeon-mimo-compressed-sensing,amiri-sgd-air}. Other angles include quantization \cite{chang-mac-quantization,tegin-low-resolution-adc-dac-arxiv}, incentive mechanisms \cite{khan-incentives}, power control \cite{zhang-power-control,liu-privacy-free} and optimization \cite{tran-time-energy,chen-convergence-time,chen-joint-learning-communications}.

\subsection{Structure of paper}

The rest of this paper is structured as follows. In Section~\ref{sec:problem-statement} we define the problem and introduce the definition of a minimax estimation scheme in this setting. We summarize our main results for achievability in Section~\ref{sec:main-results-analog-achievability} and discuss how they compare to existing digital lower bounds in Section~\ref{sec:discussion-and-comparison}. We then derive our lower bounds and discuss them in Section~\ref{sec:main-results-analog-lower-bounds}. Proofs of our results are in Sections~\ref{sec:proofs-of-analog-achievability} (achievability) and \ref{sec:proofs-of-analog-lower-bounds} (lower bounds). We validate our results with simulations in Section~\ref{sec:simulations} and conclude in Section~\ref{sec:conclusions}.

\section{Problem Statement}
\label{sec:problem-statement}

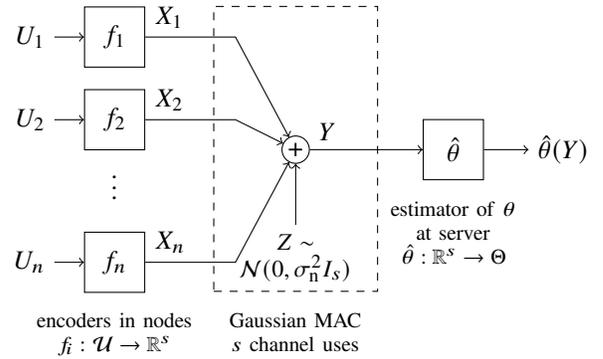
\begin{figure}
    \centering
%
%
%
\begin{tikzpicture}
    \tikzset{block/.style={draw, minimum size=8mm}, node distance=11mm}
    \node[block] (f1) {$f_1$};
    \node[block, below of=f1] (f2) {$f_2$};
    \node[below=0cm of f2] (fdots) {$\vdots$};
    \node[block, below of=fdots] (fn) {$f_n$};
    \draw[->] (f1) ++(-0.8,0) node[anchor=east] (u1 label) {$U_1$} -- (f1);
    \draw[->] (f2) ++(-0.8,0) node[anchor=east] (u2 label) {$U_2$} -- (f2);
    \draw[->] (fn) ++(-0.8,0) node[anchor=east] (un label) {$U_n$} -- (fn);
    \coordinate (midway) at ($(f2)!0.5!(fdots)$);
    \coordinate[right=1.2cm of f1] (spacer 1);
    \coordinate[right=8mm of spacer 1] (adder vertical);
    \node[circle,draw,inner sep=1pt] (adder) at (adder vertical |- midway) {+};
    \draw[->] (f1) -- (f1 -| spacer 1) -- (adder);
    \draw[->] (f2) -- (f2 -| spacer 1) -- (adder);
    \draw[->] (fn) -- (fn -| spacer 1) -- (adder);
    \node[anchor=south west] at (f1.east) {$X_1$};
    \node[anchor=south west] at (f2.east) {$X_2$};
    \node[anchor=south west] at (fn.east) {$X_n$};
    \draw[->] (adder) ++(0,-1)
        node[anchor=north, align=center, font=\small] (z label)
        {$Z \sim $ \\ $\mathcal{N}(0, \sigma_\mathrm{n}^2 I_s)$} -- (adder);
    \node[block, right=1.5cm of adder] (theta hat) {$\hat\theta$};
    \draw[->] (adder) -- (theta hat);
    \node[anchor=south west] at (adder.east) {$Y$};
    \draw[->] (theta hat) -- ++(1,0) node[anchor=west] {$\hat\theta(Y)$};
    \node[below=1ex of theta hat, anchor=north, align=center, font=\footnotesize] (estimator label) {
        estimator of $\theta$ \\
        at server \\
        $\hat\theta: \mathbb{R}^s \rightarrow \Theta$
    };

    \coordinate[left=9mm of adder] (mac left);
    \coordinate[right=9mm of adder] (mac right);
    \draw[dashed] (z label.south -| mac left) rectangle (f1.north -| mac right);
    \node[below=1ex of z label.south, anchor=north, align=center, font=\footnotesize] (mac label) {Gaussian MAC \\ $s$ channel uses};
    \node[anchor=north, align=center, font=\footnotesize] (encoder label) at (mac label.north -| fn) {encoders in nodes \\ $f_i: \mathcal{U} \rightarrow \mathbb{R}^s$};
\end{tikzpicture}
    \caption{System diagram}
    \label{fig:system-diagram}
\end{figure}

We study statistical estimation over a Gaussian multiple-access channel, a system diagram of which is in Fig.~\ref{fig:system-diagram}. In each channel use $t = 1, \dots, s$, each of $n$ senders transmits a symbol $X_{1t}, \dots, X_{nt} \in \mathbb{R}$ to the central server (which we interchangeably refer to as the \emph{receiver}), which receives a noisy superposition
\begin{equation}
    \label{eq:gaussian-mac}
    Y_t = X_{1t} + X_{2t} + \cdots + X_{nt} + Z_t,
\end{equation}
where $Z_t \sim \mathcal{N}(0, \sigma_\mathrm{n}^2)$ is the noise in the $t$th channel use. We assume an average power constraint $P$ on each sender,
\begin{equation}
    \frac1s \sum_{t=1}^s \mathbb{E} [X_{it}^2] \le P, \quad \text{for all }i = 1, \dots, n,
    \label{eq:power-constraint}
\end{equation}
where the expectation is over whatever randomness might exist in $X_{it}$, which we will make more precise shortly.

This system has the following estimation task: Each of the $n$ senders has an i.i.d.\ sample $U_i,\ i = 1, \dots, n$, from an unknown distribution $p_\theta$ on an alphabet $\mathcal{U}$, belonging to a parameterized family of distributions $\mathcal{P} = \{p_\theta: \theta \in \Theta\}$ with parameter space $\Theta \subseteq \mathbb{R}^d$. We use the notation $\mathbb{E}_\theta[\cdot]$ to mean expectation under the distribution $p_\theta$. The goal of the receiver is to estimate $\theta$ given $Y \triangleq (Y_1, \dots, Y_s)$.

To complete this task, each sender $i$ encodes its sample using a function $f_i: \mathcal{U} \rightarrow \mathbb{R}^s$ to produce $X_i \triangleq (X_{i1}, \dots, X_{is}) = f_i(U_i)$. The receiver, which knows the encoding functions, uses an estimator $\hat\theta(Y)$. We thus define how an estimation is carried out.
\begin{definition}
    An \emph{estimation scheme} for $s$ channel uses is a pair $(\mathbf{f}, \hat\theta)$ comprising $n$ encoding functions $\mathbf{f} = (f_1, \dots, f_n)$, where $f_i : \mathcal{U} \rightarrow \mathbb{R}^s$ is used by sender $i$, and an estimator function $\hat\theta : \mathbb{R}^s \rightarrow \Theta$ used by the receiver.
\end{definition}

We are now in a position to elaborate on the average power constraint in \eqref{eq:power-constraint}. The distribution of $X_i$ depends (via $f$) on $p_\theta$, which is not known in advance. We therefore require that schemes respect this power constraint for every $\theta \in \Theta$, that is, that the encoding functions $\{f_i\}$ satisfy
\begin{equation}
    \frac1s \mathbb{E}_\theta\!\left[\lVert f_i(U_i) \rVert_2^2\right] \le P, \quad \text{for all } i \in \{1, \dots, n\}, \theta \in \Theta.
    \label{eq:power-constraint-in-scheme}
\end{equation}

To evaluate possible schemes, we study risk under squared error loss, with the goal of minimizing the squared error $\mathbb{E}_\theta \lVert\hat\theta(Y) - \theta\rVert_2^2$. If we fix the encoding functions $\mathbf{f}$, all that remains is to choose an estimator function $\hat\theta$. We can understand these estimators using the same frameworks as in classical statistics; the difference is that our estimator can access only $Y$, not the samples $\{U_i\}$. In particular, when $\mathbf{f}$ is fixed, we will call an estimator \emph{minimax} if it minimizes the worst-case risk (over $\theta \in \Theta$).

In our context, it is natural to extend this idea to schemes. When referring to the risk of a scheme $(\mathbf{f}, \hat\theta)$, we mean the risk when that scheme is used. To remind ourselves that this also depends on the encoding functions $\mathbf{f}$, we write the risk as $R(\theta; \mathbf{f}, \hat\theta) = \mathbb{E}_\theta \lVert\hat\theta(Y) - \theta\rVert_2^2$, with $\mathbf{f}$ being implicit on the right-hand side. We can then extend minimaxity to schemes.

\begin{definition}
    Consider a class $\mathcal{S}$ of estimation schemes for $s$ channel uses. A scheme $(\mathbf{f}_\mathrm{M}, \hat\theta_\mathrm{M})$ is \emph{minimax} for $\mathcal{S}$ if it minimizes the maximum risk among all those schemes in $\mathcal{S}$ that also satisfy the power constraint \eqref{eq:power-constraint-in-scheme}. That is, if $\mathcal{S}_P$ is the subset of $\mathcal{S}$ satisfying $\eqref{eq:power-constraint-in-scheme}$, then a scheme $(\mathbf{f}, \hat\theta)$ is minimax if it satisfies
    \begin{equation}
        \inf_{(\mathbf{f}, \hat\theta) \in \mathcal{S}_P} \sup_\theta R(\theta; \mathbf{f}, \hat\theta) = \sup_\theta R(\theta; \mathbf{f}_\mathrm{M}, \hat\theta_\mathrm{M}).
    \end{equation}
\end{definition}

Where a scheme's encoding functions are the same for all nodes, $f_i = f$ for all $i = 1, \dots, n$, we will abuse notation by writing the common encoding function $f$ in place of the collection $\mathbf{f}$, for example, $R(\theta; f, \hat\theta) \triangleq R(\theta; \mathbf{f}, \hat\theta)$.

The main achievability results in Section~\ref{sec:main-results-analog-achievability} and the lower bound corollaries in Section~\ref{sec:main-results-analog-lower-bounds} are concerned with two cases of the general problem. The first is the \textbf{Gaussian location} model, in which $p_\theta = \mathcal{N}(\theta, \sigma^2 I_d)$, with $\mathcal{U} = \mathbb{R}^d$ and $\Theta = \{\theta \in \mathbb{R}^d: \lVert\theta\rVert_2 \le B\sqrt d\}$ for some known $B > 0$. The goal of the receiver is to estimate the mean $\theta$ of the multivariate Gaussian distribution with known covariance matrix $\sigma^2 I_d$.

The second is the \textbf{product Bernoulli parameter} model, in which $p_\theta = \prod_{j=1}^d \mathrm{Bernoulli}(\theta_j)$, with $\mathcal{U} = \{0, 1\}^d$ and $\Theta = [0, 1]^d$. The goal of the receiver is to estimate the mean $\theta$ of the Bernoulli distribution.

We note that the gradient aggregation problem in distributed stochastic gradient descent, a key part of federated machine learning, can be cast as a distributed parameter estimation problem of this type; see \textit{e.g.} \cite{barnes-rtopk}.

\section{Results for Analog Achievability}
\label{sec:main-results-analog-achievability}

We develop linear estimation schemes for the Gaussian and Bernoulli mean estimation tasks described in Section~\ref{sec:problem-statement}. A defining feature of these schemes is their \textit{analog} nature: they simply scale and transmit their samples to the central server in an uncoded fashion. This is in contrast to a digital approach where samples are encoded with a finite number of bits, which are then reliably communicated to the server using channel coding techniques. This analog approach allows us to exploit the additive nature of the Gaussian MAC to average the statistical samples over the air. The following theorems characterize the risk of these analog schemes.

\begin{theorem}
    \label{thm:gaussian-location-minimax-scheme}
     In the Gaussian location model, consider the class of all estimation schemes for $d$ channel uses, and using a scale-and-offset encoding function common to all senders $f(u) = \alpha u + \beta$ for some $\alpha \in \mathbb{R}, \beta \in \mathbb{R}^d$ (and any estimator function). The minimax scheme is given by the choice
    \begin{equation}
        \label{eq:gaussian-location-minimax-encoding-function}
        f_\mathrm{M}(u) = \sqrt{\frac{P}{B^2 + \sigma^2}} u,
        \quad
        \hat\theta_\mathrm{M}(Y) = \frac1n \sqrt{\frac{B^2 + \sigma^2}{P}} Y,
    \end{equation}
    and yields the minimax risk
    \begin{equation}
        \label{eq:gaussian-location-minimax-scheme-risk}
        \sup_\theta R(\theta; f_\mathrm{M}, \hat\theta_\mathrm{M}) = \frac{d \sigma^2}{n} \left[1 + \frac{\sigma_\mathrm{n}^2}{n P} \left(1 + \frac{B^2}{\sigma^2}\right)\right].
    \end{equation}
\end{theorem}

By using a repetition code, Theorem~\ref{thm:gaussian-location-minimax-scheme} can be extended to cases where $s > d$.

\begin{corollary}
    \label{cor:gaussian-location-repetition-scheme}
    In the Gaussian location model, if $s \ge d$, there exists a scheme $(f_\mathrm{R}, \hat\theta_\mathrm{R})$ achieving the worst-case risk
    \begin{equation}
        \label{eq:gaussian-location-repetition-scheme-risk}
        \sup_\theta R(\theta; f_\mathrm{R}, \hat\theta_\mathrm{R}) = \frac{d \sigma^2}{n} \left[1 + \frac{\sigma_\mathrm{n}^2}{\lfloor s/d \rfloor n P} \left(1 + \frac{B^2}{\sigma^2}\right)\right].
    \end{equation}
    This scheme involves repeating the encoding function \eqref{eq:gaussian-location-minimax-encoding-function} $\lfloor s/d \rfloor$ times, leaving the remaining $s - d \lfloor s/d \rfloor$ channel uses unused, and averaging the corresponding repeated estimates.
\end{corollary}

The proofs of Theorem~\ref{thm:gaussian-location-minimax-scheme} and Corollary~\ref{cor:gaussian-location-repetition-scheme} are in Section~\ref{sec:proof-achievability-gaussian-location}.

Where $s/d$ is not an integer, the unused channel uses could be filled with another partial repetition, giving a slight improvement on \eqref{eq:gaussian-location-repetition-scheme-risk} but a more unwieldy expression.

For the product Bernoulli parameter model, we provide the minimax scheme among those using affine estimators.

\begin{theorem}
    \label{thm:product-bernoulli-affine-estimator-minimax-risk}
    In the product Bernoulli parameter model, consider the class of all estimation schemes for $d$ channel uses ($s = d$), and using affine estimators. The minimax scheme in this class is the one using the encoding function defined per element
    \begin{equation}
        \label{eq:product-bernoulli-encoding-function}
        [f_\mathrm{M}(u)]_t = \begin{cases}
            -\sqrt{P}, &\text{ if }[u]_t = 0\\
            \sqrt{P},  &\text{ if }[u]_t = 1,
        \end{cases}
    \end{equation}
    where $[\cdot]_t$ is the $t$th element of its (vector) argument, and the estimator function $\hat\theta_{\mathrm{M}}(Y) = \alpha_\mathrm{M} Y + \beta_\mathrm{M} \mathbf{1}$, where $\beta_\mathrm{M} = \frac12$ and
    \begin{equation}
        \alpha_\mathrm{M} = \begin{cases}
            \dfrac{1}{2 \sqrt{n P} (\sqrt n + 1)},
                &\text{ if }\sigma_\mathrm{n}^2 \le n^{3/2} P, \\[14pt]
            \dfrac{n\sqrt{P}}{2(\sigma_\mathrm{n}^2 + n^2 P)},
                &\text{ if }\sigma_\mathrm{n}^2 \ge n^{3/2} P.
        \end{cases}
        \label{eq:bernoulli-affine-estimator-minimax-alpha}
    \end{equation}
    The minimax risk given by this choice of $f_\mathrm{M}$ and $(\alpha_\mathrm{M}, \beta_\mathrm{M})$ is
    \begin{align*}
        &\sup_\theta R(\theta; f_\mathrm{M}, \hat\theta_{\mathrm{M}}) = \\ &\qquad\qquad \begin{cases}
            \dfrac{d}{4 (\sqrt n + 1)^2}  \left(1 + \dfrac{\sigma_\mathrm{n}^2}{nP}\right), &\text{ if }\sigma_\mathrm{n}^2 \le n^{3/2} P, \\[14pt]
            \dfrac{d}{4} \cdot \dfrac{1}{1 + n \cdot \frac{n P}{\sigma_\mathrm{n}^2}}, &\text{ if }\sigma_\mathrm{n}^2 \ge n^{3/2} P.
        \end{cases}
        \numberthis
        \label{eq:product-bernoulli-affine-estimator-minimax-risk}
    \end{align*}
\end{theorem}

We can also similarly extend this using a repetition code.

\begin{corollary}
    \label{cor:product-bernoulli-repetition-scheme}
    In the product Bernoulli parameter model, if $s \ge d$, there exists a scheme $(f_\mathrm{R}, \hat\theta_\mathrm{R})$ achieving the risk
    \begin{align*}
        &\sup_\theta R(\theta; f_\mathrm{R}, \hat\theta_{\mathrm{R}}) = \\ &\qquad \begin{cases}
            \dfrac{d}{4 (\sqrt n + 1)^2}  \left(1 + \dfrac{\sigma_\mathrm{n}^2}{\lfloor s/d \rfloor nP}\right), &\text{ if }\sigma_\mathrm{n}^2 \le n^{3/2} P, \\[14pt]
            \dfrac{d}{4} \cdot \dfrac{1}{1 + n \cdot \frac{\lfloor s/d \rfloor  n P}{\sigma_\mathrm{n}^2}}, &\text{ if }\sigma_\mathrm{n}^2 \ge n^{3/2} P.
        \end{cases}
        \numberthis
        \label{eq:product-bernoulli-repetition-affine-estimator-minimax-risk}
    \end{align*}
    This scheme involves repeating the encoding function \eqref{eq:product-bernoulli-encoding-function} $\lfloor s/d \rfloor$ times, leaving the remaining $s - d \lfloor s/d \rfloor$ channel uses unused, and averaging the corresponding repeated estimates.
\end{corollary}

The proofs of Theorem~\ref{thm:product-bernoulli-affine-estimator-minimax-risk} and Corollary~\ref{cor:product-bernoulli-repetition-scheme} are in Section~\ref{sec:proof-achievability-product-bernoulli}.

\section{Comparison to Digital Lower Bounds}
\label{sec:discussion-and-comparison}
In the previous section, we characterized the performance of analog estimation schemes for the Gaussian and Bernoulli models. In this section, we compare their performance to digital approaches that have been studied in the recent literature, and show that analog schemes can lead to drastically smaller estimation error for the same amount of physical resources, \textit{i.e.}\ transmission power and number of channel uses.

In particular, recent work in machine learning \cite{duchi,garg,braverman,diakonikolas,archayaetal,barnes-lowerbounds} has studied the impact of communication constraints on distributed parameter estimation. These works abstract out the physical layer, simply assuming a constraint on the number of bits available to represent each sample. This implicitly corresponds to assuming that communication is done in a digital fashion, with channel coding used to transmit the resultant bits without any errors. For example, in \cite{barnes-lowerbounds}, the authors develop information-theoretic lower bounds on the minimax squared error risk over a parameter space $\Theta \subset \mathbb{R}^d$,
\begin{equation*}
    \sup_{\theta \in \Theta,\, \mathbf{f} \in \mathcal{F}_k^\mathrm{D}} R(\theta; \mathbf{f}, \hat\theta) = \sup_{\theta \in \Theta,\, \mathbf{f} \in \mathcal{F}_k^\mathrm{D}} \mathbb{E}_\theta \lVert\hat\theta(Y) - \theta\rVert_2^2,
\end{equation*}
where $\mathcal{F}_k^\mathrm{D}$ now is defined as the set of all possible encoding schemes $\mathbf{f} \triangleq (f_1, \dots, f_n)$, where $f_i(u) \in \{1, 2, \dots, 2^k\}$ for all $i = 1, \dots n$, \textit{i.e.}\ each sample $U_i$ is quantized to $k$ bits, which are then noiselessly communicated to the receiver. Note that these information-theoretic results lower bound the minimax risk achieved by any $k$-bit digital estimation scheme. In this section, we compare our results to these lower bounds, as applied to the Gaussian MAC we study in this paper.

We assume that senders can transmit errorlessly at the Shannon capacity of the channel.
The capacity region of a Gaussian multiple-access channel with $n$ users, power $P$ and channel noise $\sigma_\mathrm{n}^2$ is given by the region of all $(R_1, \dots, R_n)$ satisfying \cite{cover-infotheory}
\begin{equation}
    \sum_{i \in S} R_i < \frac12 \log_2 \left( 1 + \frac{|S| \bar P}{\sigma_\mathrm{n}^2} \right), \qquad \forall S \subseteq \{1, \dots, n\}.
    \label{eq:capacity-of-awgn-mac-succinct}
\end{equation}

We allocate rates equally among all the senders, in which case the inequality in which $S$ comprises all the senders dominates. If the MAC channel is utilized $s$ times, we assume that each sender is able to noiselessly communicate
\begin{equation}
    k = \sup_{(R_1, \dots, R_n)} \frac{s}{n} \sum_{i=1}^n R_i = \frac{s}{2n} \log_2 \left( 1 + \frac{n P}{\sigma_\mathrm{n}^2} \right)\qquad\text{bits}
    \label{eq:capacity-of-awgn-mac-applied}
\end{equation}
to the receiver. Note that at finite block lengths, the senders cannot communicate to the receiver at the Shannon capacity and that this optimistic assumption benefits the performance of the digital schemes.

We can then substitute \eqref{eq:capacity-of-awgn-mac-applied} into the aforementioned lower bound on minimax risk in \cite{barnes-lowerbounds}. However, that result (Corollary 5 therein) has an unspecified constant, inhibiting direct comparisons to our result from Corollary~\ref{cor:gaussian-location-repetition-scheme} above. Therefore, we instead take advantage of more recent work \cite{barnes-fisher-arxiv} to rederive the bound without the unspecified constant, and combine this with \eqref{eq:capacity-of-awgn-mac-applied} to arrive at the following digital lower bound for the Gaussian location model.

\begin{proposition}
    \label{prop:digital-lower-bound-gaussian-location-model}
    In the Gaussian location model, consider all schemes in which senders send bits to the receiver at the Shannon capacity for $s$ channel uses. For $\frac{s}{n} \log_2 \left( 1 + \frac{n P}{\sigma_\mathrm{n}^2} \right) < d$, the risk associated with any such scheme is at least
    \begin{equation}
        \sup_{\lVert\theta\rVert_2 \le B\sqrt d} \mathbb{E}_\theta \lVert\hat\theta - \theta\rVert_2^2 \ge
        \frac{d\sigma^2}{\frac{s}{d} \log_2 \left(1 + \frac{n P}{\sigma_\mathrm{n}^2}\right) + \frac{\pi^2 \sigma^2}{B^2}}.
        \label{eq:gaussian-location-barnes-lowerbound-applied}
    \end{equation}
\end{proposition}

\begin{proof}
    Use Theorem~\ref{thm:fisher-information-under-shannon-constraint} (see Section~\ref{sec:proofs-of-analog-lower-bounds}) to rederive Theorem 2 of \cite{barnes-lowerbounds} without the unspecified constant. Then follow the modifications to Corollary 1 of \cite{barnes-lowerbounds}, then Corollary 5 of \cite{barnes-lowerbounds}. Finally, substitute \eqref{eq:capacity-of-awgn-mac-applied} into this modified result. Note that $[-B,B]^d \subset \Theta \triangleq \{\theta: \lVert\theta\rVert_2 \le B\sqrt d\}$, as required by their lower bound result.
\end{proof}

Compare this to \eqref{eq:gaussian-location-repetition-scheme-risk} from Corollary~\ref{cor:gaussian-location-repetition-scheme}.
Note that this proposition implies that as the number of nodes and therefore the number of samples $n$ increases, the risk of any digital scheme decreases  as $\Omega(d^2/s \log n)$, whereas the risk of the scheme from Corollary~\ref{cor:gaussian-location-repetition-scheme} scales with $O(d/n)$. This implies that, when $s \ge d$, the analog schemes can lead to an exponentially smaller estimation error, or equivalently require an exponentially smaller number of samples to achieve the same estimation accuracy, as compared to digital schemes employing the same physical resources.

On the other hand, we make a brief note on the case where $s < d$. Here, an analog scheme transmitting scaled versions of samples cannot easily communicate more coordinates than it has channel uses. A natural approach would be for each node to transmit only the (scaled) first $s$ elements of $U_i$. In this case, the worst-case risk would scale as $\Theta(d)$, independent of $s$ and $n$, which is the maximal risk achievable even in the absence of any samples. Thus, the digital scheme achieves risk better than $\Theta(d)$ whenever $s = \omega(d/\log(1+\frac{nP}{\sigma_\mathrm{n}^2}))$, which can be the case when the SNR or $n$ is large, while our analog scheme requires $s \ge d$ to be viable.

We also have a similar situation for the Bernoulli model, derived in the same way.

\begin{proposition}
    \label{prop:digital-lower-bound-product-bernoulli-model}
    In the product Bernoulli model, consider all schemes in which senders send bits to the receiver at the Shannon capacity for $s$ channel uses. For $\frac{s}{n} \log_2 \left( 1 + \frac{n P}{\sigma_\mathrm{n}^2} \right) < d$, the risk associated with any such scheme is at least
    \begin{equation}
        \sup_{\theta \in [0,1]^d} \mathbb{E}_\theta \lVert\hat\theta - \theta\rVert_2^2 \ge
        \frac{d}{\frac{s}{d} \frac{1}{(\frac12 - 2 \varepsilon^2)^2} \log_2 \left( 1 + \frac{n P}{\sigma_\mathrm{n}^2} \right) + \frac{\pi^2}{\varepsilon^2}}
        \label{eq:product-bernoulli-barnes-lowerbound-applied}
    \end{equation}
\end{proposition}

\begin{proof}
    Apply a similarly modified version of Corollary~8 from \cite{barnes-lowerbounds} (via Corollary~4 of \cite{barnes-lowerbounds} and Theorem~2 of \cite{barnes-fisher-arxiv}), using \eqref{eq:capacity-of-awgn-mac-applied}.
\end{proof}

\begin{table*}
\centering
\caption{Comparison of results}
\label{tab:comparison-of-results}
    {\tabulinesep=1ex
    \begin{tabu}{r|c|c|c}
        & analog achievability & digital lower bound & analog lower bound \\
        \hline \hline
        Gaussian location &
        $\frac{d \sigma^2}{n} \left[1 + \frac{\sigma_\mathrm{n}^2}{\lfloor s/d \rfloor n P} \left(1 + \frac{B^2}{\sigma^2}\right)\right]$ &
        $\frac{d\sigma^2}{\frac{s}{d} \log_2 \left(1 + \frac{n P}{\sigma_\mathrm{n}^2}\right) + \frac{\sigma^2}{B^2} \pi^2}$ &
        $\frac{d\sigma^2}{n} \cdot \frac{1}{\frac{s}{d} \log_2\left(1 + \frac{n P}{\sigma_\mathrm{n}^2}\right) + \frac{\sigma^2}{B^2} \frac{\pi^2}{n}}$ \\
        \hline
        product Bernoulli &
        $\frac{d}{4 (\sqrt n + 1)^2}  \left(1 + \frac{\sigma_\mathrm{n}^2}{\lfloor s/d \rfloor nP}\right)$ &
        $\frac{d}{\frac{s}{d} \frac{1}{(\frac12 - 2 \varepsilon^2)^2} \log_2 \left( 1 + \frac{n P}{\sigma_\mathrm{n}^2} \right) + \frac{\pi^2}{\varepsilon^2}}$ &
        $\frac{d}{n} \cdot \frac{1}{\frac{s}{d} \frac{1}{(\frac12 - 2\varepsilon^2)^2} \log_2\left(1 + \frac{n P}{\sigma_\mathrm{n}^2}\right) + \frac{\pi^2}{n\varepsilon^2}}.$ \\
        \hline
        $s \ge d$ (both) &
        $O\!\left(\dfrac{d}{n}\right)$ &
        $\Omega\!\left(\dfrac{d^2}{s \log n}\right)$ &
        $\Omega\!\left(\dfrac{d^2}{s \cdot n \log n}\right)$ \\
        \hline
        $s \propto d$ (both) &
        $O\!\left(\dfrac{d}{n}\right)$ &
        $\Omega\!\left(\dfrac{d}{\log n}\right)$ &
        $\Omega\!\left(\dfrac{d}{n \log n}\right)$ \\
        \hline
    \end{tabu}
    }
\end{table*}

Note that analogously to the Gaussian case, this result implies that the risk of any digital scheme for the Bernoulli model scales as $\Omega(d^2/s \log n)$, while the risk of our analog scheme from Corollary~\ref{cor:product-bernoulli-repetition-scheme} decreases as $O(d/n)$. In the low SNR case, $\sigma_\mathrm{n}^2 \ge n^{3/2} P$, the analog scheme appears to achieve $O(d^2/sn^2)$ (compared to $\Omega(d/n)$), but since increasing $n$ also increases the SNR given by $nP/\sigma_\mathrm{n}^2$, this relationship will eventually give way to the high SNR regime, where $\sigma_\mathrm{n}^2 \le n^{3/2} P$. These results show that building on the inherent summation of transmitted signals in the Gaussian MAC to perform the averaging that classical statistical estimators would do can provide drastic gains in estimation performance. Similar gains have been observed in asymptotic lossy source coding for Gaussian sensor networks in \cite{gastpar-TIT-2008,gastpar-IPSN-2003}, where one is interested in communicating an i.i.d.\ Gaussian source over a MAC under mean-squared error distortion, as well as for distributed stochastic gradient descent in experimental comparisons in \cite{amiri-sgd-air}.

We summarize all these results, as well as the new lower bounds we present in Section~\ref{sec:main-results-analog-lower-bounds} below, in Table~\ref{tab:comparison-of-results}.

\section{Results for Analog Lower Bounds}
\label{sec:main-results-analog-lower-bounds}

So far in this paper, we have shown that the analog joint communication-estimation schemes in Section~\ref{sec:main-results-analog-achievability} can exponentially outperform lower bounds for digital schemes with a physical layer abstraction. Having shown this striking advantage, the next question this naturally raises is what lower bounds exist for analog schemes. That is: if these analog schemes can achieve what digital schemes could never hope for, how close are we to the fundamental limits for such analog approaches?

In this section, we present new lower bounds on worst-case squared error risk for any estimation scheme in our setting of interest where the parametric model has a sub-Gaussian score function. This work draws on the findings of \cite{barnes-fisher-arxiv}, for which we first introduce some relevant quantities. We then apply them to the same models for which we presented achievability results in Section~\ref{sec:main-results-analog-achievability}. We find that these lower bounds are within a logarithmic factor of the risk achieved by the schemes presented above.

\subsection{Preliminaries}
\label{sec:analog-lower-bounds-preliminaries}

If $U \sim p_\theta$, where $p_\theta$ is a member of a family of probability distributions parameterized by $\theta \in \Theta \subseteq \mathbb{R}^d$ and differentiable in $\theta$, the \emph{score function} is defined as the gradient of the log-likelihood function,
\begin{align*}
    S_\theta(u)
    &\triangleq \nabla_\theta \log p_\theta(u) \\
    &= \left( \frac{\partial}{\partial \theta_1} \log p_\theta(u), \dots, \frac{\partial}{\partial \theta_d} \log p_\theta(u) \right).
\end{align*}
Where we have many samples $U_1, \dots, U_n \sim p_\theta$, we may denote the score function of the finite sequence as
\begin{equation*}
    S_\theta(u_1, \dots, u_n) \triangleq \nabla_\theta \log p_\theta(u_1, \dots, u_n).
\end{equation*}
Note that both $S_\theta(u)$ and $S_\theta(u_1, \dots, u_n)$ have the same number of elements as $\theta$, independent of the size of the argument passed into $S_\theta(\cdot)$. It is a well-known property of the score function that $\mathbb{E}[S_\theta(U)] = 0$.

The \emph{Fisher information} is then defined as the $d \times d$ matrix
\begin{equation*}
    I_U(\theta) \triangleq \mathbb{E}[S_\theta(U) S_\theta(U)^\mathsf{T}]
\end{equation*}
which makes its trace equal to
\begin{equation*}
    \mathbf{tr}(I_U(\theta)) = \sum_{j=1}^d \mathbb{E}\!\left[\left(\frac{\partial}{\partial \theta_j} \log p_\theta(u)\right)^2\right].
\end{equation*}

We say that a zero-mean random variable $X$ is \emph{sub-Gaussian with parameter $\rho$} if
\begin{equation*}
    \mathbb{E}\left[\exp(\lambda X)\right] \le \exp\left( \frac{\lambda^2 \rho^2}{2} \right)
    \quad \text{for all }\lambda \in \mathbb{R}.
\end{equation*}

Recall that if a zero-mean random variable $X$ is bounded within $[a, b]$ with probability 1 then it is sub-Gaussian with parameter $(b-a)/2$. Also, if $X_1, \dots, X_n$ are independent and sub-Gaussian with parameters $\rho_1, \dots, \rho_n$, then their sum $X_1 + \cdots + X_n$ is sub-Gaussian with parameter $\sqrt{\rho_1^2 + \cdots + \rho_n^2}$.

Our lower bounds require the regularity conditions described in \cite{barnes-fisher-arxiv}, which we recite here:
\begin{enumerate}[label=(\roman*)]
    \item \label{itm:regularity-density-differentiable}
    $\sqrt{p_\theta(u_1, \dots, u_n)}$ is continuously differentiable with respect to each component $\theta_j$ at almost all $(u_1, \dots, u_n) \in \mathcal{U}^n$ (with respect to some measure dominating ${\{p_\theta: \theta \in \Theta\}}$).
    \item \label{itm:regularity-fisher-continuous}
    The Fisher information for each component $\theta_j$, $\mathbb{E}([\frac{\partial}{\partial \theta_j} \log p_\theta(U_1, \dots, U_n) ]^2)$, exists and is a continuous function of $\theta_j$.
    \item \label{itm:regularity-channel-squareintegrable}
    The conditional density $p(y|u_1, \dots, u_n)$ is square integrable in the sense that for almost all $y \in \mathbb{R}^s$ for each $\theta$, $\int p(y|u_1, \dots, u_n)^2 \,dp_\theta(u_1, \dots, u_n) < \infty$.
\end{enumerate}
It is easily verified that \ref{itm:regularity-density-differentiable} and \ref{itm:regularity-fisher-continuous} are satisfied in both the Gaussian location and product Bernoulli parameter models. As for \ref{itm:regularity-channel-squareintegrable}, this follows from the bounded conditional density $p(y|x_1, \dots, x_n)$ of the Gaussian MAC; details are in the relevant proofs.

\subsection{Main lower bound on worst-case risk}
\label{sec:main-analog-lower-bound}

We have now covered the background necessary to state our main lower bound, which is on the squared error risk for any estimation scheme in the setting described in Section~\ref{sec:problem-statement} when the parametric model has a sub-Gaussian score.

\begin{theorem}
    \label{thm:main-lower-bound}
    Suppose that $[-B, B]^d \subset \Theta$, and that the samples $(U_i)_{i=1}^n$ are i.i.d.\ and satisfy conditions \ref{itm:regularity-density-differentiable} and \ref{itm:regularity-fisher-continuous}, and that $\langle v, S_\theta(U_i) \rangle$ is sub-Gaussian with parameter $\rho$ for all unit vectors $v \in \mathbb{R}^d$.
    Then in a Gaussian multiple-access channel with $s$ channel uses, the worst-case risk under squared error loss of any estimation scheme $(\mathbf{f}, \hat\theta)$ must satisfy
    \begin{equation}
        \sup_{\theta \in \Theta} R(\theta; \mathbf{f}, \hat\theta) \ge \frac{d}{n} \cdot \frac{1}{\frac{s}{d} \rho^2 \log_2\left(1 + \frac{n P}{\sigma_\mathrm{n}^2}\right) + \frac{\pi^2}{n B^2}}.
    \end{equation}
\end{theorem}

The proof for this builds on a relationship between Fisher information and Shannon information established by \cite{barnes-fisher-arxiv}, and a Bayesian Cramer-Rao type bound to relate minimax risk and Fisher information.
We provide the proof in Section~\ref{sec:proof-of-main-lower-bound}.

\subsection{Bounds for specific models}
\label{sec:analog-lower-bounds-specific-models}

In the case of the Gaussian location model, we can characterize this bound in terms of the sample variance $\sigma^2$.

\begin{corollary}
    \label{cor:gaussian-location-lower-bound}
    In the Gaussian location model with $s$ channel uses, the worst-case risk under squared error loss of any estimation scheme $(\mathbf{f}, \hat\theta)$ must satisfy
    \begin{equation}
        \sup_{\theta \in \Theta} R(\theta; \mathbf{f}, \hat\theta) \ge \frac{d\sigma^2}{n} \cdot \frac{1}{\frac{s}{d} \log_2\left(1 + \frac{n P}{\sigma_\mathrm{n}^2}\right) + \frac{\sigma^2}{B^2} \frac{\pi^2}{n}}.
    \end{equation}
\end{corollary}

We can also derive a result for product Bernoulli models where elements of $\theta$ are close to $\frac12$, \textit{i.e.}, where the samples are dense.

\begin{corollary}
    \label{cor:product-bernoulli-lower-bound}
    Consider the relatively dense product Bernoulli model, where $U_1, \dots, U_n \sim \prod_{j=1}^d \mathrm{Bernoulli}(\theta_j)$, with $\Theta = [\frac12 - \varepsilon, \frac12 + \varepsilon]^d$, $\varepsilon \in (0, \frac12)$, with $s$ channel uses. The worst-case risk under squared error loss of any estimation scheme $(\mathbf{f}, \hat\theta)$ must satisfy
    \begin{equation}
        \sup_{\theta \in \Theta} R(\theta; \mathbf{f}, \hat\theta) \ge \frac{d}{n} \cdot \frac{1}{\frac{s}{d} \frac{1}{(\frac12 - 2\varepsilon^2)^2} \log_2\left(1 + \frac{n P}{\sigma_\mathrm{n}^2}\right) + \frac{\pi^2}{n\varepsilon^2}}.
    \end{equation}
\end{corollary}

The proofs of the above two corollaries, which both follow from Theorem~\ref{thm:main-lower-bound}, are in Section~\ref{sec:proof-lower-bound-specific-problem-instances}.

\subsection{Lower bound for a general multiple-access channel}
\label{sec:lower-bound-general-mac}

The result of Theorem~\ref{thm:main-lower-bound} can be generalized to other multiple-access channels, in terms of the total capacity of the network, that is, the maximum achievable sum of all rates in the network,
\begin{equation}
    C_\mathrm{total} = \max_{\prod_i p_i(x_{it})} I(X_{1t}, \dots, Y_{nt}; Y_t),
\end{equation}
where the maximum is over all product distributions for $(X_{1t}, \dots, X_{nt})$. We will refer to this quantity as the ``sum capacity'', recalling that it does not fully describe the capacity region of the network.

\begin{theorem}
    \label{thm:general-bound}
    Suppose that the samples $(U_i)_{i=1}^n$ are i.i.d.\ and satisfy conditions \ref{itm:regularity-density-differentiable} and \ref{itm:regularity-fisher-continuous}, and that $\langle v, S_\theta(U_i) \rangle$ is sub-Gaussian with parameter $\rho$ for all unit vectors $v \in \mathbb{R}^d$. Consider any discrete memoryless multiple-access channel that is constrained by the sum capacity $C_\mathrm{total}$ (per channel use), whose conditional density $p(y|x_1, \dots, x_n)$ is bounded. The worst-case risk of any estimation scheme $(\mathbf{f}, \hat\theta)$ must satisfy
    \begin{equation}
        \sup_{\theta \in \Theta} R(\theta; \mathbf{f}, \hat\theta) \ge \frac{d}{n} \cdot \frac{1}{2\frac{s}{d} \rho^2 C_\mathrm{total}+ \frac{\pi^2}{nB^2}}.
    \end{equation}
\end{theorem}

\begin{proof}
    Follow the proof of Theorem~\ref{thm:main-lower-bound}, but replace the right-hand side of \eqref{eq:main-result-gaussian-mac-shannon-information} with $sC_\mathrm{total}$ (\textit{i.e.}, $C_\mathrm{total}$ for $s$ channel uses). Note that the stipulation that $p(y|x_1, \dots, x_n)$ be bounded (by some different finite $M$) ensures that \eqref{eq:gaussian-mac-satisfies-regularity-channel-squareintegrable}, and hence \ref{itm:regularity-channel-squareintegrable}, is satisfied.
\end{proof}

\section{Proofs of Analog Achievability}
\label{sec:proofs-of-analog-achievability}

\subsection{Gaussian location model}
\label{sec:proof-achievability-gaussian-location}

In this section, we prove our main results for the Gaussian location model, Theorem~\ref{thm:gaussian-location-minimax-scheme} and Corollary~\ref{cor:gaussian-location-repetition-scheme}. In this model, the samples $U_i \sim \mathcal{N}(\theta, \sigma^2 I_d)$, where $\theta$ lies in an $\ell_2$-ball in a $d$-dimensional space, $\{\lVert\theta\rVert_2 \le B\sqrt d\}$, and the goal is to estimate $\theta$.

Since the multiple-access channel already produces a sum, one might suspect that that an estimation scheme emulating the sample mean would be a natural candidate, given its properties in classical estimation. Indeed, it is minimax among schemes using affine encoders (and any estimator). First, we show that this estimator is minimax for a fixed affine encoder, as we state formally in the following proposition.

\begin{proposition}
    \label{prop:gaussian-location-minimax-estimator-for-scale-and-offset-encoder}
    In the Gaussian location model, let the senders use any scale-and-offset encoding function $f(u) = \alpha u + \beta$  for some $\alpha \in \mathbb{R}, \beta \in \mathbb{R}^d$, common to all senders, and assume that this encoding function satisfies the power constraint, and that the channel is used $d$ times (\textit{i.e.}, $s = d$). Then the minimax estimator is given by
    \begin{equation*}
        \hat\theta_\mathrm{M}(Y) = \frac{1}{\alpha n} Y - \frac{1}{\alpha}\beta,
    \end{equation*}
    which yields risk
    \begin{equation}
        \label{eq:gaussian-location-minimax-risk}
        \mathbb{E}_\theta \lVert\hat\theta_\mathrm{M}(Y) - \theta\rVert_2^2 = \frac{d}{n} \left(\sigma^2 + \frac{\sigma_\mathrm{n}^2}{n \alpha^2}\right).
    \end{equation}
\end{proposition}

\begin{remark}
    The estimator given by Proposition~\ref{prop:gaussian-location-minimax-estimator-for-scale-and-offset-encoder} is also the maximum likelihood estimator.
\end{remark}

The proof for this follows similar lines to the classical result using a least favorable sequence of priors, with modifications for channel noise.

\begin{lemma}
    \label{lem:gaussian-location-bayes-estimator}
    If $\boldsymbol\theta$ is distributed according to the prior $\mathcal{N}(\mu, b^2 I_d)$, and all senders use the common encoding function $f(u) = \alpha u + \beta$ for some $\alpha \in \mathbb{R}, \beta \in \mathbb{R}^d$, then the Bayes estimator $\hat\theta_{\mu,b^2}(y)$ is given by
    \begin{equation}
        \hat\theta_{\mu,b^2}(Y) = \mu + \frac{\alpha n b^2}{\alpha^2 n^2 b^2 + \alpha^2 n \sigma^2 + \sigma_\mathrm{n}^2} (Y - \alpha n \mu + n \beta),
    \end{equation}
    and the Bayes risk is
    \begin{equation}
        \mathbb{E} \lVert \hat\theta_{\mu,b^2}(Y) - \boldsymbol\theta \rVert^2 = \frac{d (\alpha^2 n \sigma^2 + \sigma_\mathrm{n}^2)}{\alpha^2 n^2 + \frac{\alpha^2 n \sigma^2 + \sigma_\mathrm{n}^2}{b^2}}.
    \end{equation}
\end{lemma}

\begin{proof}
    Under squared error loss, the Bayes estimator for $\mathcal{N}(\mu, b^2 I_d)$ is (by well-known theorem, \textit{e.g.}\ \cite[Cor.\ 4.1.2(a)]{lehmann-tpe}) $\hat\theta_{\mu, b^2}(y) = \mathbb{E}(\boldsymbol\theta|y)$, which we will evaluate. The relevant covariance matrices are
    \begin{align*}
        \Sigma_Y &= (\alpha^2 n^2 b^2 + \alpha^2 n \sigma^2 + \sigma_\mathrm{n}^2) I_d, \\
        \Sigma_{Y\boldsymbol\theta}
        &= \mathbb{E}\!\left[\left(\alpha n W + \alpha \textstyle\sum_i V_i\right)  W^\top\right]
        = \alpha n b^2 I_d.
    \end{align*}
    Then the Bayes estimate is given by
    \begin{align*}
        &\hat\theta_{\mu,b^2}(Y)
        = \mathbb{E}(\boldsymbol\theta|Y)
        = \mathbb{E}\boldsymbol\theta + \Sigma_{\boldsymbol\theta Y}\Sigma_Y^{-1}(Y - \mathbb{E} Y) \\
        &\qquad = \mu + \alpha n b^2 \cdot \frac{1}{\alpha^2 n^2 b^2 + \alpha^2 n \sigma^2 + \sigma_\mathrm{n}^2} \cdot (Y - \alpha n \mu - n \beta),
    \end{align*}
    and since this estimator is unbiased, the squared error is given by the trace of the conditional variance,
    \begin{align*}
        \mathbb{E}\lVert\hat\theta_{\mu,b^2} - \boldsymbol\theta\rVert^2
        &= \mathbf{tr}\,\mathbf{var}(\boldsymbol\theta | Y)
        = \mathbf{tr}(\Sigma_{\boldsymbol\theta} - \Sigma_{\boldsymbol\theta Y}\Sigma_Y^{-1}\Sigma_{Y\boldsymbol\theta}) \\
        &= db^2 - \frac{d(\alpha n b^2)^2}{\alpha^2 n^2 b^2 + \alpha^2 n \sigma^2 + \sigma_\mathrm{n}^2}. \qedhere
    \end{align*}
\end{proof}

\begin{proof}[Proof of Proposition~\ref{prop:gaussian-location-minimax-estimator-for-scale-and-offset-encoder}]
    Take the Bayes estimator from Lemma~\ref{lem:gaussian-location-bayes-estimator}. Let $b^2 \rightarrow \infty$, then we have a sequence of priors $\mathcal{N}(\mu, b^2)$ yielding increasing Bayes risk converging to
    \begin{align*}
        \lim_{b \rightarrow \infty} \mathbb{E}\lVert\hat\theta_{\mu,b^2} - \boldsymbol\theta\rVert^2 = \frac{d(\alpha^2 n \sigma^2 + \sigma_\mathrm{n}^2)}{\alpha^2 n^2} = \frac{d}{n} \left(\sigma^2 + \frac{\sigma_\mathrm{n}^2}{n\alpha^2}\right).
    \end{align*}
    The minimax estimator is then
    \begin{equation*}
        \lim_{b \rightarrow \infty} \hat\theta_{\mu,b^2}(Y) = \mu + \frac{1}{\alpha n} Y - \mu - \frac{1}{\alpha}\beta. \qedhere
    \end{equation*}
\end{proof}

In the absence of a power constraint, the offset $\beta$ has no effect---since it is known, it is easily cancelled by the receiver's estimator. Intuitively, with a power constraint, one would expect no offset to be preferable. In Theorem~\ref{thm:gaussian-location-minimax-scheme}, where we find the best choice of $(\alpha, \beta)$, we find that this is indeed the case.

\begin{proof}[Proof of Theorem~\ref{thm:gaussian-location-minimax-scheme}]
    For any given $\alpha, \beta$, the minimax risk from Proposition~\ref{prop:gaussian-location-minimax-estimator-for-scale-and-offset-encoder} is decreasing in $\alpha$. Therefore, we choose the largest $\alpha$ satisfying the power constraint \eqref{eq:power-constraint-in-scheme}. Note that
    \begin{align*}
        \mathbb{E}_\theta \!\left[ \lVert X_i \rVert_2^2 \right]
        &= \alpha^2 (\lVert\theta\rVert^2 + d\sigma^2) + 2\alpha\theta^\top\beta + \lVert\beta\rVert^2 \\
        &= \lVert \alpha \theta + \beta \rVert_2^2 + \alpha^2 d\sigma^2.
        \numberthis
    \end{align*}
    We thus solve
    \begin{equation}
        \label{eq:gaussian-location-power-constraint-optimization-problem}
        \begin{array}{rll}
            \text{maximize} & \alpha \\
            \text{subject to} & \lVert \alpha \theta + \beta \rVert_2^2 + \alpha^2 d\sigma^2 \le dP &\forall \theta: \lVert\theta\rVert \le \sqrt d B.
        \end{array}
    \end{equation}
    If we relax the constraint to $\lVert \beta \pm \alpha B \mathbf{1} \rVert_2^2 + \alpha^2 d\sigma^2 \le dP$, we can use Lagrange multipliers to find the solution
    \begin{equation}
        \alpha = \sqrt{\frac{P}{B^2 + \sigma^2}}, \quad \beta = 0,
    \end{equation}
    and verify that it also satisfies the constraints of, and is therefore also a solution to, \eqref{eq:gaussian-location-power-constraint-optimization-problem}.
\end{proof}

We now turn to the case where $s > d$. A natural extension of the scheme from Theorem~\ref{thm:gaussian-location-minimax-scheme} would be to transmit repetitions of the sample.

\begin{lemma}
    \label{lem:repetition-schemes-for-affine-estimators}
    Let $(\mathbf{f}, \hat\theta)$ be a scheme with $\hat\theta(Y)$ affine in $Y$ and consider a scheme $(\mathbf{f}_R, \hat\theta_\mathrm{R})$ that repeats the encoding function $m$ times and averages the estimates for each repetition, $\hat\theta_\mathrm{R}(Y) = \frac1m \sum_{j=1}^m \hat\theta([Y]_j)$, where $[Y]_j$ is the part of $Y$ corresponding to the $j$th repetition. The risk of $(\mathbf{f}_R, \hat\theta_\mathrm{R})$ is the same as for $(\mathbf{f}, \hat\theta)$, but with $\sigma_\mathrm{n}^2/m$ in place $\sigma_\mathrm{n}^2$.
\end{lemma}

\begin{proof}
    The bias of the estimator is unaffected by the repetition (and is independent of $\sigma_\mathrm{n}^2$), and if the original estimator is written as $\hat\theta(Y) = AY + c$, the variance can be shown to be $\sum_i \mathbf{var}(A X_i) + \frac1m \mathbf{var}(AZ)$. Relative to the original estimator variance, this is equivalent to dividing $\sigma_\mathrm{n}^2 I$ by $m$.
\end{proof}

This then yields the achievability result of Corollary~\ref{cor:gaussian-location-repetition-scheme}.
\begin{proof}[Proof of Corollary~\ref{cor:gaussian-location-repetition-scheme}]
    Apply Lemma~\ref{lem:repetition-schemes-for-affine-estimators} to Theorem~\ref{thm:gaussian-location-minimax-scheme}, with $m = \lfloor s/d \rfloor$ and ignoring the leftover channel uses.
\end{proof}

Comparing this to the $s = d$ case, the repetition reduces the noise by a factor of roughly $s/d$, which is the expected effect of averaging a repeated transmission. The minimax risk then converges more quickly to the noiseless case as $s/d \rightarrow \infty$.

\subsection{Product Bernoulli model}
\label{sec:proof-achievability-product-bernoulli}

In this section, we prove our main results for the Bernoulli parameter model, Theorem~\ref{thm:product-bernoulli-affine-estimator-minimax-risk} and Corollary~\ref{cor:product-bernoulli-repetition-scheme}. In this model, $U_i \sim \prod_{i=1}^d \mathrm{Bernoulli}(\theta)$, and the goal is to estimate $\theta$, which is in $[0, 1]^d$. Our calculations in this section will work with the parameterized encoding function common to all senders
\begin{equation}
    \label{eq:bernoulli-parameter-encoding-function}
    f_C(u) = \begin{cases}
        -C, &\text{ if }u = 0\\
        C,  &\text{ if }u = 1.
    \end{cases}
\end{equation}
Our analysis of the Bernoulli parameter model focuses on the scalar case, as stated in Proposition~\ref{prop:scalar-bernoulli-minimax-affine-scheme} below. Theorem~\ref{thm:product-bernoulli-affine-estimator-minimax-risk} will then follow by extension to independent dimensions.

\begin{proposition}
    \label{prop:scalar-bernoulli-minimax-affine-scheme}
    In the scalar Bernoulli parameter model ($d = 1$), consider the class of all estimation schemes using affine estimators $\hat\theta_{\alpha,\beta}(Y) = \alpha Y + \beta, \alpha, \beta \in \mathbb{R}$ (and any scalar encoding function with $s=1$). The minimax scheme in this class is the one using the encoding function
    \begin{equation}
        f_\mathrm{M}(u) = \begin{cases}
            -\sqrt{P}, &\text{ if }u = 0\\
            \sqrt{P},  &\text{ if }u = 1,
        \end{cases}
    \end{equation}
    and the estimator $\hat\theta_{\mathrm{M}}(Y) = \alpha_\mathrm{M} Y + \beta_\mathrm{M}$, where $\beta_\mathrm{M} = \frac12$ and $\alpha_\mathrm{M}$ is as provided in \eqref{eq:bernoulli-affine-estimator-minimax-alpha}. The minimax risk given by this choice of $(\alpha_\mathrm{M}, \beta_\mathrm{M})$ is
    \begin{equation}
        \sup_\theta R(\theta; f_\mathrm{M}, \hat\theta_{\mathrm{M}}) = \begin{cases}
            \frac{1}{4 (\sqrt n + 1)^2} \left(1 + \frac{\sigma_\mathrm{n}^2}{nP}\right), &\text{ if }\sigma_\mathrm{n}^2 \le n^{3/2} P, \\[6pt]
            \frac{1}{4} \cdot \frac{1}{1 + n \cdot {n P}/{\sigma_\mathrm{n}^2}}, &\text{ if }\sigma_\mathrm{n}^2 \ge n^{3/2} P.
        \end{cases}
        \label{eq:bernoulli-affine-estimator-minimax-risk}
    \end{equation}
\end{proposition}

Our steps for proving Proposition~\ref{prop:scalar-bernoulli-minimax-affine-scheme} will be first to establish the minimax risk for the common encoding function $f_C$, then to show that a scheme using any other encoding function can be transformed to one using $f_C$ for some $C$ of equal risk, and finally to show that the optimal value for $C$ is $\sqrt P$. Before we continue, we compute the risk for a general affine estimator.

\begin{lemma}
    In the scalar Bernoulli parameter model ($d = 1$), if all senders use the encoding function $f_C$ \eqref{eq:bernoulli-parameter-encoding-function}, and the receiver uses the affine estimator $\hat\theta_{\alpha, \beta}(Y) = \alpha Y + \beta$,
    then the risk is
    \begin{align*}
        R(\theta; f_C, \hat\theta_{\alpha, \beta})
        &= \alpha^2 \left[ 4nC^2 \theta(1 - \theta) + \sigma_\mathrm{n}^2 \right] \\
        &\qquad + \left[ \alpha n C (2\theta - 1) + \beta - \theta \right]^2.
        \numberthis
        \label{eq:bernoulli-affine-estimator-risk}
    \end{align*}
\end{lemma}

\begin{proof}
    Recall that $Y = \sum_{i=1}^n f_C(U_i) + Z$ and that $f_C(U_i) = C$ w.p.\ $\theta$  and $f_C(U_i) = -C$ w.p.\ $1 - \theta$.
    The variance and bias of the estimator are then
    \begin{align*}
        \mathbf{var}_\theta[\hat\theta_{\alpha, \beta}(Y)]
        &= \alpha^2 \mathbf{var}(Y)
        = \alpha^2 \left[ 4nC^2 \theta(1 - \theta) + \sigma_\mathrm{n}^2 \right]. \\
        \mathbf{bias}_\theta[\hat\theta_{\alpha, \beta}(Y)] &\triangleq \mathbb{E}_\theta \hat\theta_{\alpha, \beta}(Y) - \theta = \alpha n C (2\theta - 1) + \beta - \theta.
    \end{align*}
    The result then follows from combining these as
    \begin{equation*}
        \mathbb{E}_\theta[\hat\theta_{\alpha, \beta}(Y) - \theta]_2^2 = \mathbf{var}_\theta[\hat\theta_{\alpha, \beta}(Y)] + (\mathbf{bias}_\theta[\hat\theta_{\alpha, \beta}(Y)])^2. \qedhere
    \end{equation*}
\end{proof}

The bulk of the work in proving Proposition~\ref{prop:scalar-bernoulli-minimax-affine-scheme} is in showing Proposition~\ref{prop:bernoulli-minimax-affine-estimator}, which establishes the minimax estimator for the encoding function $f_C$.

\begin{proposition}
    \label{prop:bernoulli-minimax-affine-estimator}
    In the scalar Bernoulli parameter model, let all senders use the encoding function $f_C(u)$ from \eqref{eq:bernoulli-parameter-encoding-function}, and consider the class of all affine estimators $\hat\Theta_\mathrm{aff} = \{\hat\theta_{\alpha,\beta}(Y) = \alpha Y + \beta, \alpha, \beta \in \mathbb{R}\}$. The minimax affine estimator is given by $\hat\theta_{\mathrm{M}}(Y) = \alpha_\mathrm{M} Y + \beta_\mathrm{M}$, where $\beta_\mathrm{M} = \frac12$ and
    \begin{equation}
        \alpha_\mathrm{M} = \begin{cases}
            \frac{1}{2 \sqrt{n} C (\sqrt n + 1)},
                &\text{ if }\sigma_\mathrm{n}^2 \le n^{3/2} C^2, \\[6pt]
            \frac{nC}{2(\sigma_\mathrm{n}^2 + n^2 C^2)},
                &\text{ if }\sigma_\mathrm{n}^2 \ge n^{3/2} C^2.
        \end{cases}
        \label{eq:bernoulli-affine-estimator-symmetric-encoder-minimax-alpha}
    \end{equation}
    The minimax risk given by this choice of $(\alpha_\mathrm{M}, \beta_\mathrm{M})$ is
    \begin{equation}
        \sup_\theta R(\theta; f_\mathrm{M}, \hat\theta_\mathrm{M}) = \begin{cases}
            \frac{1}{4 (\sqrt n + 1)^2} \left(1 + \frac{\sigma_\mathrm{n}^2}{nC^2}\right), &\text{ if }\sigma_\mathrm{n}^2 \le n^{3/2} C^2, \\[6pt]
            \frac{1}{4} \cdot \frac{1}{1 + n \cdot {n C^2}/{\sigma_\mathrm{n}^2}}, &\text{ if }\sigma_\mathrm{n}^2 \ge n^{3/2} C^2.
        \end{cases}
        \label{eq:bernoulli-affine-estimator-symmetric-encoder-minimax-risk}
    \end{equation}
\end{proposition}

\begin{proof}
    Define $\alpha_\mathrm{lo} = \frac{1}{2 \sqrt{n} C (\sqrt n + 1)}$ and $\alpha_\mathrm{hi} = \frac{n C}{2(\sigma_\mathrm{n}^2 + n^2 C^2)}.$
    Note that then $\alpha_\mathrm{M} = \min\{\alpha_\mathrm{lo}, \alpha_\mathrm{hi}\}$. For convenience, and with some abuse of notation, let $R(\theta; f_C, \alpha, \beta)$ refer to the expression in \eqref{eq:bernoulli-affine-estimator-risk}. We will repeatedly use the facts that:
    \begin{enumerate}[label=(\alph*)]
        \item $R(\theta; f_C, \alpha_\mathrm{lo}, \beta_\mathrm{M})$ is constant with respect to $\theta$.
        \label{itm:bernoulli-parameter-low-noise-risk-constant}
        \item $R(\theta; f_C, \alpha_\mathrm{hi}, \beta_\mathrm{M})$ is convex in $\theta$ and minimized at $\theta \in \{0, 1\}$, at which the risk is equal.
        \label{itm:bernoulli-parameter-high-noise-risk-convex}
        \item $R(0; f_C, \alpha, \beta_\mathrm{M})$ is convex in $\alpha$ and minimized at $\alpha = \alpha_\mathrm{hi}$.
        \label{itm:bernoulli-parameter-high-noise-risk-convex-in-alpha}
    \end{enumerate}
    \vspace{-2ex}
    These can all be verified by appropriate substitutions into \eqref{eq:bernoulli-affine-estimator-risk}. Where we invoke these facts, we will label the equality or inequality signs accordingly.

    We will show that for every other choice $(\alpha, \beta)$, there exists some $\theta \in [0, 1]$ with risk exceeding $\sup_\theta R(\theta; \alpha_\mathrm{M}, \beta_\mathrm{M})$. We divide into three cases.

    \textit{Case 1:} $\alpha > \alpha_\mathrm{lo}$, or $\alpha = \alpha_\mathrm{lo}$ and $\beta \ne \frac12$. In this case, take $\theta = \frac12$ and we have
    \begin{align*}
        R(\tfrac12; f_C, \alpha, \beta)
        &= \alpha^2 ( nC^2 + \sigma_\mathrm{n}^2 ) + (\beta - \tfrac12)^2 \\
        &> \alpha_\mathrm{lo}^2 ( nC^2 + \sigma_\mathrm{n}^2 ) = R(\tfrac12; f_C, \alpha_\mathrm{lo}, \beta_\mathrm{M}).
    \end{align*}
    Then, if $\alpha_\mathrm{M} = \alpha_\mathrm{lo}$, then by \ref{itm:bernoulli-parameter-low-noise-risk-constant}, the right-hand side is equal to $\sup_\theta R(\theta; f_C, \alpha_\mathrm{M}, \beta_\mathrm{M})$. If $\alpha_\mathrm{M} = \alpha_\mathrm{hi}$, then note that
    \begin{align*}
        R(\tfrac12; f_C, \alpha_\mathrm{lo}, \beta_\mathrm{M})
        &\labelsign{a}= R(0; f_C, \alpha_\mathrm{lo}, \beta_\mathrm{M}) \\
        &\labelsign{c}\ge R(0; f_C, \alpha_\mathrm{hi}, \beta_\mathrm{M})
        \labelsign{b}= \sup_\theta R(\theta; f_C, \alpha_\mathrm{M}, \beta_\mathrm{M}),
    \end{align*}
    where labeled steps refer to corresponding facts above.

    \textit{Case 2:} $\alpha < \alpha_\mathrm{lo}$ and $\beta \ge \frac12$. Take $\theta = 0$ and we have
    \begin{align*}
        R(0; f_C, \alpha, \beta)
        &= \alpha^2 \sigma_\mathrm{n}^2 + ( \beta - \alpha n C )^2 \\
        &\ge \alpha^2 \sigma_\mathrm{n}^2 + ( \tfrac12 - \alpha n C )^2
        = R(0; f_C, \alpha, \beta_\mathrm{M}),
    \end{align*}
    where in the inequality we used the fact that $\alpha n C < \alpha_\mathrm{lo} n C = \frac{\sqrt n}{2(\sqrt n + 1)} < \frac12$. Then, if $\alpha_\mathrm{M} = \alpha_\mathrm{lo}$, we also have $\alpha_\mathrm{lo} < \alpha_\mathrm{hi}$, and by fact \ref{itm:bernoulli-parameter-high-noise-risk-convex-in-alpha}, is strictly decreasing in $\alpha$ for all $\alpha < \alpha_\mathrm{lo}$, thus
    \begin{equation*}
        R(0; f_C, \alpha, \beta_\mathrm{M}) > R(0; f_C, \alpha_\mathrm{lo}, \beta_\mathrm{M}) \labelsign{a}= \sup_\theta R(\theta; f_C, \alpha_\mathrm{M}, \beta_\mathrm{M}).
    \end{equation*}
    If $\alpha_\mathrm{M} = \alpha_\mathrm{hi}$, then we have
    \begin{equation*}
        R(0; f_C, \alpha, \beta_\mathrm{M})
        \labelsign{c}\ge R(0; f_C, \alpha_\mathrm{hi}, \beta_\mathrm{M})
        \labelsign{b}= \sup_\theta R(\theta; f_C, \alpha_\mathrm{M}, \beta_\mathrm{M}).
    \end{equation*}

    \textit{Case 3:} $\alpha < \alpha_\mathrm{lo}$ and $\beta \le \frac12$. Take $\theta = 1$ and argue similarly to case 2 that $R(1; f_C, \alpha, \beta_\mathrm{M}) > \sup_\theta R(\theta; f_C, \alpha_\mathrm{M}, \beta_\mathrm{M})$.
\end{proof}

\begin{lemma}
    \label{lem:bernoulli-parameter-model-offset-does-nothing}
    In the scalar Bernoulli parameter model, consider the scheme $(f, \hat\theta)$, in which all senders use the encoding function $f(0) = A,\ f(1) = B$, and the receiver uses the estimator $\hat\theta$. Then there exists a scheme $(f', \hat\theta')$ satisfying $f'(0) = -f'(1)$ and with minimax risk equal to that of $(f, \hat\theta)$.
\end{lemma}

\begin{proof}
    Choose $C = \frac{B - A}{2}$, so that $f'(u) \triangleq f_C(u) = f(u) - \frac{A+B}{2}$. By construction, $f'(0) = -f'(1) = \frac{A - B}{2}$. Then, if $Y$ and $Y'$ are what the receiver observes under $f$ and $f'$ respectively, we have
    $
        Y'
        = \sum_i f'(U_i) + Z
        = \sum_i [f(U_i) - \tfrac{A+B}{2}] + Z
        = Y - n\tfrac{A+B}{2}.
    $
    We can then define $\hat\theta'(Y') \triangleq \hat\theta(Y' + n\tfrac{A+B}{2})$, and this will have exactly the same statistical properties as $\hat\theta(Y)$.
\end{proof}

Now we may complete the proof of Proposition~\ref{prop:scalar-bernoulli-minimax-affine-scheme}.

\begin{proof}[Proof of Proposition~\ref{prop:scalar-bernoulli-minimax-affine-scheme}]
    Because Lemma~\ref{lem:bernoulli-parameter-model-offset-does-nothing} shows there is no sacrifice in minimax risk, it suffices to consider just schemes using encoding functions of the form $f_C$ in \eqref{eq:bernoulli-parameter-encoding-function}. The minimax affine estimator for such encoding functions is found in Proposition~\ref{prop:bernoulli-minimax-affine-estimator}. From \eqref{eq:bernoulli-affine-estimator-symmetric-encoder-minimax-risk}, the minimax risk for $f_C$ is strictly decreasing in $C^2$. Therefore, to minimize over all encoding functions $f_C$, we take the highest-magnitude $C$ satisfying the power constraint \eqref{eq:power-constraint-in-scheme}, $C = \pm \sqrt P$.
\end{proof}

The extension of this result to the product Bernoulli model is then an application of the scalar case on a per-sample basis.

\begin{proof}[Proof of Theorem~\ref{thm:product-bernoulli-affine-estimator-minimax-risk}]
    Because each dimension $1, \dots, d$ is independent, each dimension can be optimized separately. Each sender transmits its $j$th sample $[f_C(U_i)]_j$ using the scheme from Proposition~\ref{prop:scalar-bernoulli-minimax-affine-scheme}. The even division of power still satisfies the average power constraint \eqref{eq:power-constraint-in-scheme}. The minimax risk is then $d$ times the minimax risk along one dimension.
\end{proof}

Finally, Corollary~\ref{cor:product-bernoulli-repetition-scheme} follows using Lemma~\ref{lem:repetition-schemes-for-affine-estimators} again.

\begin{proof}[Proof of Corollary~\ref{cor:product-bernoulli-repetition-scheme}]
    Apply Lemma~\ref{lem:repetition-schemes-for-affine-estimators} to Theorem~\ref{thm:product-bernoulli-affine-estimator-minimax-risk}, with $m = \lfloor s/d \rfloor$ and ignoring the leftover channel uses.
\end{proof}

\section{Proofs of Analog Lower Bounds}
\label{sec:proofs-of-analog-lower-bounds}

\subsection{Preliminaries}
\label{sec:proofs-of-analog-lower-bounds-preliminaries}

We first present two results that are key to our main theorem. These results characterize the worst-case risk in terms of the trace of the Fisher information matrix, and in turn in terms of Shannon information. First, Equation 8 of \cite{barnes-lowerbounds} tells us the following, which we list as a lemma here.

\begin{lemma}
    \label{lem:barnes-squared-error-risk}
    Suppose $[-B, B]^d \subset \Theta$. For any estimator $\hat\theta(Y_1, \dots, Y_n)$, the worst-case squared error risk must satisfy
    \begin{equation}
        \sup_{\theta \in \Theta} \mathbb{E} \lVert \hat\theta(Y) - \theta \rVert^2 \ge \frac{d^2}{\sup_{\theta \in \Theta} \mathbf{tr}(I_Y(\theta)) + \frac{d\pi^2}{B^2}}.
    \end{equation}
\end{lemma}

We will also lean on the following theorem, due to \cite{barnes-fisher-arxiv}.

\begin{theorem}
    \label{thm:fisher-information-under-shannon-constraint}
    Suppose that $\langle u, S_\theta(X) \rangle$ is sub-Gaussian with parameter $\rho$ for any unit vector $u \in \mathbb{R}^d$. Under regularity conditions \ref{itm:regularity-density-differentiable}--\ref{itm:regularity-channel-squareintegrable}, $\mathbf{tr}(I_Y(\theta)) \le 2\rho^2 I_\theta(X; Y)$.
\end{theorem}

\subsection{Proof of main lower bound}
\label{sec:proof-of-main-lower-bound}

We first provide an upper bound for the mutual information between the channel input and output.
For brevity we omit the proof, which can be derived using standard results in information theory.

\begin{proposition}
    \label{prop:gaussian-mac-shannon-information-weak-bound-independent-inputs}
    Consider the Gaussian multiple-access channel with $s$ channel uses, $Y = X_1 + \cdots + X_n + Z$, $Z \sim \mathcal{N}(0, \sigma_\mathrm{n}^2I_s)$, with a power constraint $\frac1s \mathbb{E}[\lVert X_i \rVert^2] \le P$. If $X_1, \dots, X_n$ are independent, then the Shannon information between its input $(X_1, \dots, X_n)$ and its output $Y$ is bounded by
    \begin{equation}
        I(X_1, \dots, X_n; Y) \le \frac{s}{2} \log_2\left(1 + \frac{n P}{\sigma_\mathrm{n}^2}\right).
    \end{equation}
\end{proposition}

We now have all of the ingredients necessary to prove Theorem~\ref{thm:main-lower-bound}, which uses the data processing inequality to chain the above results together.

\begin{proof}[Proof of Theorem~\ref{thm:main-lower-bound}]
    Recall that $X_i = f(U_i), i = 1, \dots, n$ and $Y$ is the output of the channel $p_{Y|X}(y|x_1, \dots, x_n)$ with inputs $X_1, \dots, X_n$. The conditional distribution of $Y$ given $U$ can be expressed in terms of the channel's conditional distribution,
    \begin{equation}
        \label{eq:u-to-y-is-a-channel}
        p_{Y|U}(y|u_1, \dots, u_n) = p_{Y|X}(y|f(u_1), \dots, f(u_n)).
    \end{equation}
    That is, we have a channel from $U$ to $Y$. (Note that this doesn't require invertibility in $f$, since it is in the condition, and $p_{Y|X}$ is defined by assumption.) To verify that this ``channel'' satisfies regularity condition \ref{itm:regularity-channel-squareintegrable}, note that $p_{Y|X}$ is bounded,
    \begin{equation}
        \label{eq:gaussian-mac-conditional-density-bounded}
        p_{Y|X}(y|x_1, \dots, x_n) \le \frac{1}{\sqrt{(2\pi\sigma_\mathrm{n}^2)^n}} \triangleq M,
    \end{equation}
    so chaining \eqref{eq:u-to-y-is-a-channel} and \eqref{eq:gaussian-mac-conditional-density-bounded} verifies that
    \begin{align*}
        &\int p_{Y|U}(y|u_1, \dots, u_n)^2 \,dp_U(u_1, \dots, u_n) \\
        &\qquad\le \int M^2 \,dp_U(u_1, \dots, u_n) = M^2 < \infty.
        \numberthis \label{eq:gaussian-mac-satisfies-regularity-channel-squareintegrable}
    \end{align*}
    We therefore satisfy the requirements to invoke Theorem~\ref{thm:fisher-information-under-shannon-constraint}, so long as we can establish that $\langle v, S_\theta(U_1, \dots, U_n) \rangle$ is sub-Gaussian for all unit vectors $v \in \mathbb{R}^d$. Note that since $U_1, \dots, U_n$ are independent,
    \begin{equation*}
        S_\theta(U_1, \dots, U_n)
        = \sum_{i=1}^n \nabla_\theta \log p_\theta(U_i) = \sum_{i=1}^n S_\theta(U_i).
    \end{equation*}
    Then for every unit vector $v \in \mathbb{R}^d$,
    \begin{equation*}
        \langle v, S_\theta(U_1, \dots, U_n) \rangle = \left\langle v, \sum_{i=1}^n S_\theta(U_i) \right\rangle = \sum_{i=1}^n \langle v, S_\theta(U_i) \rangle.
    \end{equation*}
    This is a sum of $n$ independent sub-Gaussian random variables each with parameter $\rho$, and is therefore sub-Gaussian with parameter $\sqrt n \rho$. Theorem~\ref{thm:fisher-information-under-shannon-constraint} thus gives
    \begin{equation}
        \label{eq:main-result-fisher-shannon}
        \mathbf{tr}(I_Y(\theta)) \le 2 n\rho^2 I_\theta(U_1, \dots, U_n; Y).
    \end{equation}
    Since $(U_1, \dots, U_n) \rightarrow (X_1, \dots, X_n) \rightarrow Y$ form a Markov chain, the data processing inequality implies that
    \begin{equation}
        \label{eq:main-result-data-processing}
        I_\theta(U_1, \dots, U_n; Y) \le I_\theta(X_1, \dots, X_n; Y).
    \end{equation}
    Now, $U_1, \dots, U_n$ are independent (by definition), and each $X_i, i = 1, \dots, n$ is a function of the corresponding $U_i$. Therefore, $X_1, \dots, X_n$ are also independent, and from Proposition~\ref{prop:gaussian-mac-shannon-information-weak-bound-independent-inputs}, we have
    \begin{equation}
        \label{eq:main-result-gaussian-mac-shannon-information}
        I(X_1, \dots, X_n; Y) \le \frac{s}{2} \log_2\left(1 + \frac{n P}{\sigma_\mathrm{n}^2}\right).
    \end{equation}
    Putting \eqref{eq:main-result-fisher-shannon}, \eqref{eq:main-result-data-processing} and \eqref{eq:main-result-gaussian-mac-shannon-information} together yields
    \begin{equation}
        \mathbf{tr}(I_Y(\theta)) \le n\rho^2s \log_2\left(1 + \frac{n P}{\sigma_\mathrm{n}^2}\right).
    \end{equation}
    Substituting this expression into the result given by Lemma~\ref{lem:barnes-squared-error-risk} then yields
    \begin{equation*}
        \sup_{\theta \in \Theta} \mathbb{E} \lVert \hat\theta(Y) - \theta \rVert^2 \ge \frac{d^2}{n\rho^2s \log_2\left(1 + \frac{n P}{\sigma_\mathrm{n}^2}\right) + \frac{d\pi^2}{B^2}}. \qedhere
    \end{equation*}
\end{proof}

\subsection{Specific problem instances}
\label{sec:proof-lower-bound-specific-problem-instances}

To find lower bounds for the Gaussian location and product Bernoulli parameter models, we compute the sub-Gaussian parameters of their score functions and apply our main result.

\begin{proof}[Proof of Corollary~\ref{cor:gaussian-location-lower-bound}]
    The score function for a single sample $U_i$ is
    \begin{align*}
        S_\theta(u_i)
        &= \nabla_\theta \left[\frac{(u_i-\theta)^\mathsf{T} (u_i-\theta)}{2\sigma^2} - \log 2\pi\sigma\right] \\
        &= \frac{1}{\sigma^2} (u_i - \theta).
    \end{align*}
    Then, with $U_i \sim \mathcal{N}(\theta, \sigma^2 I_d)$, the score function of each sample $S_\theta(U_i)$ is Gaussian with mean zero and covariance $\frac{1}{\sigma^2} I_d$. It follows that for any unit vector $v$ and each sample $U_i$, $\langle v, S_\theta(U_i) \rangle$ is Gaussian with zero mean and variance $v^\mathsf{T} \frac{1}{\sigma^2} I_d v = \frac{1}{\sigma^2} v^\mathsf{T} v = \frac{1}{\sigma^2}$. This is sub-Gaussian with parameter $\frac{1}{\sigma}$, enabling an application of Theorem~\ref{thm:main-lower-bound}.
\end{proof}

\begin{proof}[Proof of Corollary~\ref{cor:product-bernoulli-lower-bound}]
    We can compute the score function of each sample in the product Bernoulli model to be $S_\theta(u_i) = (S_{\theta_1}(u_i), \dots, S_{\theta_j}(u_i))$, where
    \begin{equation*}
        S_{\theta_j}(u_i) = \begin{cases}
            \frac{1}{\theta_j}, &\text{ if }u_{ij} = 1\\
            -\frac{1}{1-\theta_j}, &\text{ if }u_{ij} = 0.
        \end{cases}
    \end{equation*}
    Then $S_{\theta_j}(U_i)$ is bounded, and therefore is sub-Gaussian with parameter
    \begin{equation*}
        \frac{1}{2}\left[\frac{1}{\theta_j} + \frac{1}{1 - \theta_j}\right] = \frac{1}{2\theta_j (1 - \theta_j)} \le \frac{1}{\frac12 - 2\varepsilon^2},
    \end{equation*}
    where the last step uses the fact that $\theta \in \Theta = [\frac12 - \varepsilon, \frac12 + \varepsilon]^d$.

    Being the sum of $n$ independent sub-Gaussians, for all unit vectors $v \in \mathbb{R}^d$, $\langle v, S_\theta(U_i) \rangle$ is sub-Gaussian with parameter
    \begin{equation}
        \sqrt{\sum_{j=1}^d v_j^2 \frac{1}{(\frac12 - 2\varepsilon^2)^2}} = \frac{1}{\frac12 - 2\varepsilon^2}.
    \end{equation}
    This gives us a value for $\rho$ to use in Theorem~\ref{thm:main-lower-bound}.

    For $B$, we may reparameterize the parameter space to $\Theta' = [-\varepsilon, \varepsilon]$ (so that the Bernoulli component means are $\theta = \theta' + \frac12$). We can then apply Theorem~\ref{thm:main-lower-bound} to arrive at Corollary~\ref{cor:product-bernoulli-lower-bound}.
\end{proof}

\section{Simulations}
\label{sec:simulations}

\begin{figure}
    \includegraphics[width=\columnwidth]{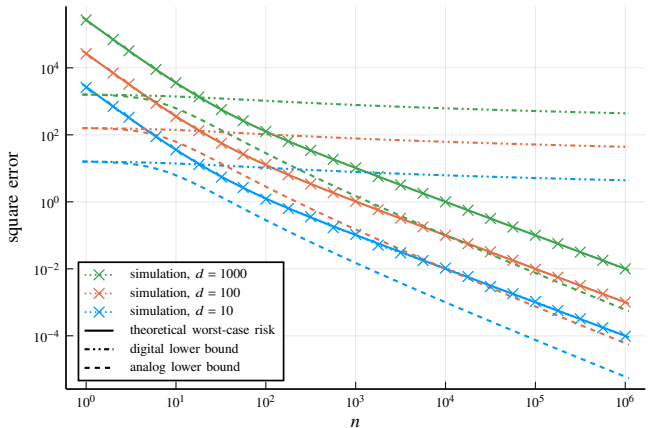}
    \caption{Simulation of Gaussian location model with $\frac{P}{\sigma_\mathrm{n}^2} = 0.1, \frac{B}{\sigma^2} = 0.4$}
    \label{fig:simulation-glm-low-snr}
\end{figure}

\begin{figure}
    \includegraphics[width=\columnwidth]{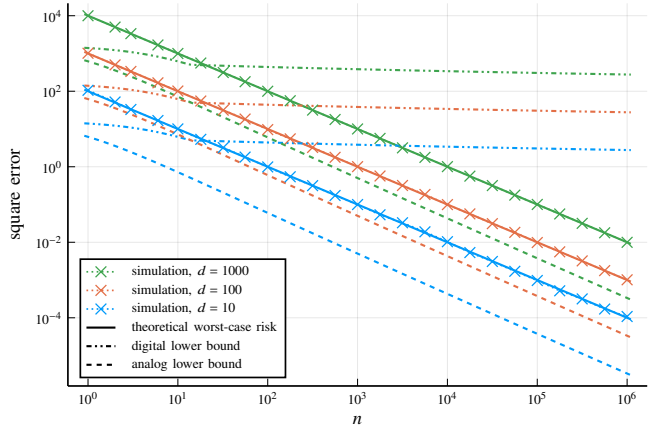}
    \caption{Simulation of Gaussian location model with $\frac{P}{\sigma_\mathrm{n}^2} = 1000, \frac{B}{\sigma^2} = 0.4$}
    \label{fig:simulation-glm-high-snr}
\end{figure}

\begin{figure}
    \includegraphics[width=\columnwidth]{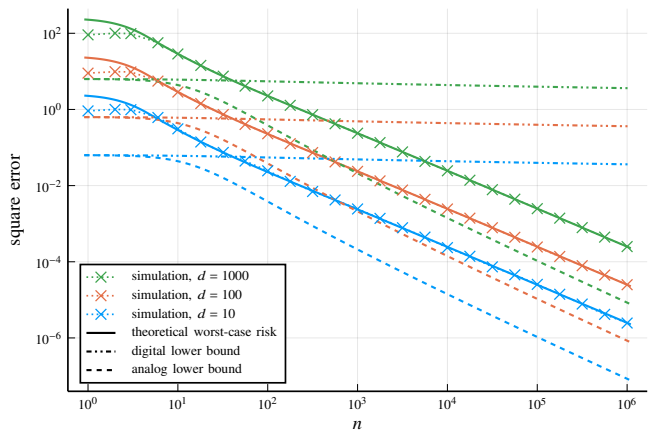}
    \caption{Simulation of product Bernoulli model with $\frac{P}{\sigma_\mathrm{n}^2} = 0.1$}
    \label{fig:simulation-pbm-low-snr}
\end{figure}

\begin{figure}
    \includegraphics[width=\columnwidth]{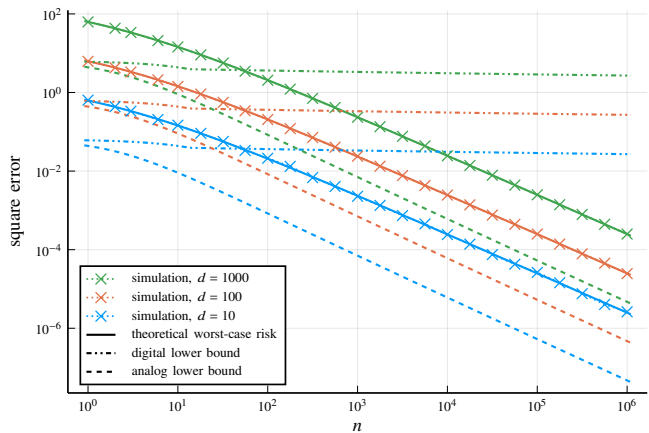}
    \caption{Simulation of product Bernoulli model with $\frac{P}{\sigma_\mathrm{n}^2} = 1000$}
    \label{fig:simulation-pbm-high-snr}
\end{figure}

We ran simulations to validate the schemes proposed in Section~\ref{sec:main-results-analog-achievability} for both of our models of interest, and we plot the results alongside each of our three theoretical results in Figs.~\ref{fig:simulation-glm-low-snr}--\ref{fig:simulation-pbm-high-snr}. Lines in the same color relate to the same value of $d$, with different patterns corresponding to different results.

The points with `$\times$' markers are the squared error of our proposed analog schemes, averaged over 100 simulations with $\theta$ drawn uniformly at random from the parameter space $\Theta$, at various values of $d$ and $n$, in the case where $s = d$. The solid lines are the theoretical results we presented in Section~\ref{sec:main-results-analog-achievability}, namely Theorems~\ref{thm:gaussian-location-minimax-scheme} and \ref{thm:product-bernoulli-affine-estimator-minimax-risk}, which are the worst-case risk over parameters $\theta \in \Theta$ for the schemes we proposed (which are minimax among the classes of estimators noted in those theorems).

Since the theoretical results are the \emph{worst-case} risk over the parameters $\theta \in \Theta$, we expect them to be at least as large as the \emph{average} squared error found in simulations with randomly sampled $\theta \in \Theta$. It turns out that in most cases, the worst-case and average risk are equal, because the risk of the minimax scheme is constant with respect to $\theta$. The one exception is the low-SNR regime ($\sigma_\mathrm{n}^2 > n^{3/2} P$, see Theorem~\ref{thm:product-bernoulli-affine-estimator-minimax-risk}) of the product Bernoulli model estimator (Fig.~\ref{fig:simulation-pbm-low-snr}), where the risk depends on $\theta$, leading to the observed difference between the theoretical line and the average over randomly drawn $\theta$ at low $n$.

The digital lower bounds of Propositions~\ref{prop:digital-lower-bound-gaussian-location-model} and \ref{prop:digital-lower-bound-product-bernoulli-model} are plotted in dash-dot-dot lines. It is here that we see the marked improvement discussed in Section~\ref{sec:discussion-and-comparison}. Since the error in digital schemes scales at best with $1/\log n$, the simulated analog schemes rapidly become significantly advantageous even in moderately large values of $n$.

Finally, the analog lower bounds of Corollaries~\ref{cor:gaussian-location-lower-bound} and \ref{cor:product-bernoulli-lower-bound} are plotted in dashed lines, and run a $\log n$ factor from the achievability and simulation results.

Recall that in the product Bernoulli model, both lower bounds depend on $\varepsilon$, where $\Theta = [\frac12 - \varepsilon, \frac12 + \varepsilon]^d$. Our achievability results and simulations, on the other hand, are for the full parameter space $\Theta = [0, 1]^d$. Note, however, that the lower bounds for any $\varepsilon < \frac12$, yielding $\Theta \subset [0, 1]^d$, also imply a lower bound for $[0, 1]^d$. To generate the plots in Figs.~\ref{fig:simulation-pbm-low-snr} and \ref{fig:simulation-pbm-high-snr}, we used $\varepsilon = \frac14$.

\section{Conclusions}
\label{sec:conclusions}

We introduced and studied a new model for minimax parameter estimation over the Gaussian multiple-access channel, developing estimation schemes for the Gaussian location model and product Bernoulli model. These ``analog'' estimation schemes directly leverage the superposition property of the Gaussian MAC, and our analysis of their risk under squared error loss showed that they exponentially outperform even lower bounds on the risk of ``digital'' schemes that separate the communication and estimation problems. We then derived new ``analog'' lower bounds for this estimation problem that are within a logarithmic factor of our achievability results. We confirmed our findings in simulations for both models. This adds theoretical insight to a growing body of literature on the advantages of analog schemes in over-the-air learning and inference, demonstrating that even fundamental limits of digital schemes can be beaten when estimation and communication are considered jointly.

This opens a number of further questions to be examined in future work. First, while the lower bounds we presented are general, our achievability results pertain specifically to the two estimation models we studied. With the additive nature of the Gaussian MAC, one would imagine that this extends to other mean estimation problems, but it raises the question of whether other estimation problems, such as distribution or empirical frequency estimation \cite{barnes-lowerbounds,han-geometric-lowerbounds,chen-trilemma}, can also harness this or other channels to achieve similarly startling gains over digital schemes.

Moreover, our estimation schemes work only when there is at least one channel use available for each parameter ($s \ge d$). In many applications, models may have too many parameters for this to be feasible, so the setting where $s < d$ is also important to study. Developing analog schemes for this regime that outperform digital approaches remains as future work.

\newcommand{\noopsort}[1]{}

\end{document}